\documentclass[conference]{ieeeconf}  

\IEEEoverridecommandlockouts                              

\title{\LARGE \bf
On Robustness Metrics for Learning STL Tasks
}

\author{Peter Varnai and Dimos V. Dimarogonas$^{1}$
\thanks{This work was partially supported by the Wallenberg AI, Autonomous Systems and Software Program (WASP) funded by the Knut and Alice Wallenberg Foundation, the Swedish Research Council (VR), the SSF COIN project, and the EU H2020 Co4Robots project.}
\thanks{$^{1}$Both authors are with the Division of Decision and Control Systems, School of Electrical Engineering and Computer Science, KTH Royal Institute of Technology, 114 28 Stockholm, Sweden. {\tt\small varnai@kth.se} (P. Varnai), {\tt\small dimos@kth.se} (D. V. Dimarogonas)}%
}

\usepackage{graphicx, subcaption}
\usepackage{caption}

\usepackage{amsmath, amssymb, amsthm}
\usepackage{acro}
\usepackage{cite}

\usepackage{pgfplots}
\usepackage{pgfplotstable}
\usepgfplotslibrary{groupplots}
\usetikzlibrary{pgfplots.groupplots}
\usetikzlibrary{matrix,positioning}
\usepgfplotslibrary{fillbetween}

\usepackage{algorithm,algorithmic}
\usepackage{enumerate}
\usepackage{xcolor}
\usepackage{multirow}
\usepackage{pifont}

\newtheoremstyle{bfplain}{}{}{\itshape}{}{\bfseries}{.}{ }{\thmname{#1}\thmnumber{ #2}\thmnote{ (#3)}}
\newtheoremstyle{bfdefinition}{}{}{}{}{\bfseries}{.}{ }{\thmname{#1}\thmnumber{ #2}\thmnote{ (#3)}}
\newtheoremstyle{itdefinition}{}{}{}{}{\itshape}{.}{ }{\thmname{#1}\thmnumber{ #2}\thmnote{ (#3)}}

\newcommand{\satisfies}{\vDash}
\renewcommand{\and}{\wedge}

\NewDocumentCommand{\until}{oo}{
	\IfNoValueTF{#1}{U}{U_{[#1,#2]}}
}
\NewDocumentCommand{\eventually}{oo}{
	\IfNoValueTF{#1}{F}{F_{[#1,#2]}}
}
\NewDocumentCommand{\always}{oo}{
	\IfNoValueTF{#1}{G}{G_{[#1,#2]}}
}

\newcommand{\diff}{\mathrm{d}}

\newcommand{\tp}{^{\textsc{T}}}

\renewcommand{\iff}{\Leftrightarrow}

\newcommand{\subnorm}[1]{_{#1}}
\NewDocumentCommand{\norm}{som}{
	\IfBooleanTF {#1}{\IfNoValueTF{#2} {\| #3 \|}{\| #3 \|\subnorm{#2}}}
	{\IfNoValueTF{#2} {	\left\| #3 \right\|}{\left\| #3 \right\|\subnorm{#2}}}
}
\NewDocumentCommand{\inReal}{soo}{
	\IfBooleanTF {#1}{}{\in}
	\mathbb{R}
	\IfNoValueTF{#2}{}{\IfNoValueTF{#3}{^{#2}}{^{#2 \times #3}}}	
}
\NewDocumentCommand{\inComplex}{soo}{
	\IfBooleanTF {#1}{}{\in}
	\mathbb{C}
	\IfNoValueTF{#2}{}{\IfNoValueTF{#3}{^{#2}}{^{#2 \times #3}}}	
}

\newcommand{\submat}[1]{_{\scriptstyle{#1}}}
\newcommand{\subvec}[1]{_{\scriptstyle{#1}}}
\newcommand{\supmat}[1]{^{\scriptstyle{#1}}}
\newcommand{\supvec}[1]{^{\scriptstyle{#1}}}

\newcommand{\mat}[1]{\mathbf{#1}}
\renewcommand{\vec}[1]{\boldsymbol{#1}}

\newcommand{\defvec}[2] {
	\DeclareDocumentCommand{#1}{s d<> d[] d''} {
		\IfBooleanTF {##1}
			{\IfNoValueTF{##2}{#2}{##2{#2}}}
			{\IfNoValueTF{##2}{\vec{#2}}{##2{\vec{#2}}}}		
		\IfNoValueTF{##3}{}{\subvec{##3}}
		\IfNoValueTF{##4}{}{\supvec{##4}}
	}
}

\newcommand{\defmat}[2] {
	\DeclareDocumentCommand{#1}{s d<> d[] d''} {
		\IfBooleanTF {##1}
			{\IfNoValueTF{##2}{#2}{[##2{#2}]}}
			{\IfNoValueTF{##2}{\mat{#2}}{##2{\mat{#2}}}}			
		\IfNoValueTF{##3}{}{\submat{##3}}
		\IfNoValueTF{##4}{}{\supmat{##4}}
	}
}
\defmat{\A}{A}
\defmat{\B}{B}
\defmat{\C}{C}
\defmat{\D}{D}
\defmat{\E}{E}
\defmat{\F}{F}
\defmat{\G}{G}
\defmat{\H}{H}
\defmat{\I}{I}
\defmat{\J}{J}
\defmat{\K}{K}
\defmat{\L}{L}
\defmat{\M}{M}
\defmat{\P}{P}
\defmat{\Q}{Q}
\defmat{\R}{R}
\defmat{\S}{S}
\defmat{\T}{T}
\defmat{\U}{U}
\defmat{\V}{V}
\defmat{\W}{W}
\defmat{\X}{X}
\defmat{\Y}{Y}
\defmat{\Z}{Z}
\defmat{\NMAT}{0}
\defmat{\SIG}{\Sigma}
\defmat{\LAM}{\Lambda}

\defvec{\ones}{1}
\defvec{\a}{a}
\defvec{\b}{b}
\defvec{\d}{d}
\defvec{\e}{e}
\defvec{\f}{f}
\defvec{\g}{g}
\defvec{\h}{h}
\defvec{\k}{k}
\defvec{\l}{l}
\defvec{\m}{m}
\defvec{\n}{n}
\defvec{\q}{q}
\defvec{\p}{p}
\defvec{\s}{s}
\defvec{\t}{t}
\defvec{\u}{u}
\defvec{\v}{v}
\defvec{\w}{w}
\defvec{\x}{x}
\defvec{\y}{y}
\defvec{\z}{z}
\defvec{\nvec}{0}
\defvec{\lam}{\lambda}
\defvec{\nuvec}{\nu}
\defvec{\pivec}{\pi}
\defvec{\phivec}{\phi}
\defvec{\rhovec}{\rho}
\defvec{\sigvec}{\sigma}
\defvec{\thetavec}{\theta}
\defvec{\alphavec}{\alpha}
\defvec{\gammavec}{\gamma}
\defvec{\Gammavec}{\Gamma}

\theoremstyle{bfdefinition}

\newtheorem{theorem}{Theorem}
\newtheorem{proposition}{Proposition}

\newtheorem{property}{Property}

\theoremstyle{itdefinition}
\newtheorem{remark}{Remark}

\colorlet{darkgreen}{green!75!black}
\colorlet{lightblue}{blue!66!white}
\colorlet{lightred}{red!77!white}
\colorlet{darkpurple}{red!50!blue!75!black}

\newcommand{\cmark}{\textcolor{darkgreen}{\ding{51}}}%
\newcommand{\xmark}{{\color{red}\ding{55}}}%

\newcommand{\NEG}{\mathcal{N}}

\NewDocumentCommand{\AND}{o}{
	\IfNoValueTF{#1}{\mathcal{A}}{\mathcal{A}_{#1}}
}
\NewDocumentCommand{\OR}{o}{
	\IfNoValueTF{#1}{\mathcal{O}}{\mathcal{O}_{#1}}
}
\NewDocumentCommand{\EVENTUALLY}{oo}{
	\IfNoValueTF{#1}{\mathcal{F}}{\mathcal{F}_{[#1,#2]}}
}
\NewDocumentCommand{\ALWAYS}{oo}{
	\IfNoValueTF{#1}{\mathcal{G}}{\mathcal{G}_{[#1,#2]}}
}
\DeclareAcronym{TL}{
	short = TL,
	long = temporal logic
}
\DeclareAcronym{LTL}{
	short = LTL,
	long = linear temporal logic
}
\DeclareAcronym{TLTL}{
	short = TLTL,
	long = truncated linear temporal logic
}
\DeclareAcronym{MITL}{
	short = MITL,
	long = metric interval temporal logic
}
\DeclareAcronym{MTL}{
	short = MTL,
	long = metric temporal logic
}
\DeclareAcronym{STL}{
	short = STL,
	long = signal temporal logic
}
\DeclareAcronym{PPC}{
	short = PPC,
	long = prescribed performance control
}
\DeclareAcronym{PI2}{
	short = {PI$^2$},
	long = policy improvement with path integrals
}
\DeclareAcronym{GPI2}{
	short = {G-PI$^2$},
	long = guided policy improvement with path integrals
}
\DeclareAcronym{AGPI2}{
	short = {AG-PI$^2$},
	long = adaptive guided policy improvement with path integrals
}
\DeclareAcronym{TLPS}{
	short = TLPS,
	long = temporal logic policy search
}
\DeclareAcronym{ReLU}{
	short = ReLU,
	long = rectified linear unit
}
\DeclareAcronym{RL}{
	short = RL,
	long = reinforcement learning
}
\DeclareAcronym{MPC}{
	short = MPC,
	long = model predictive control
}
\DeclareAcronym{MDP}{
	short = MDP,
	long = Markov decision process
}
\DeclareAcronym{MPNN}{
	short = MPNN,
	long = message passing neural network
}
\DeclareAcronym{GNN}{
	short = GNN,
	long = graph neural network
}
\DeclareAcronym{MILP}{
	short = MILP,
	long = mixed-integer linear program
}
\DeclareAcronym{HJB}{
	short = HJB,
	long = Hamilton-Jacobi-Bellman
}
\DeclareAcronym{QP}{
	short = QP,
	long = quadratic program
}
\DeclareAcronym{LP}{
	short = LP,
	long = linear program
}
\DeclareAcronym{REPS}{
	short = REPS,
	long = relative entropy policy search
}
\DeclareAcronym{SAT}{
	short = SAT,
	long = Boolean satisfiability problem
}

\begin{document}

\maketitle
\thispagestyle{empty}
\pagestyle{empty}

\begin{abstract}

Signal temporal logic (STL) is a powerful tool for describing complex behaviors for dynamical systems. Among many approaches, the control problem for systems under STL task constraints is well suited for learning-based solutions, because STL is equipped with robustness metrics that quantify the satisfaction of task specifications and thus serve as useful rewards. In this work, we examine existing and potential robustness metrics specifically from the perspective of how they can aid such learning algorithms. We show that various desirable properties restrict the form of potential metrics, and introduce a new one based on the results. The effectiveness of this new robustness metric for accelerating the learning procedure is demonstrated through an insightful case study.

\end{abstract}

\section{INTRODUCTION}

Formal methods have much potential in the field of robotics due to the ability of temporal logics to express rich and complex desired system behaviors. In accordance, developing control methods for achieving these behaviors has become an area of increasing research interest. Different temporal logics, such as \ac{LTL} \cite{pnueli1977temporal}, \ac{MTL} \cite{koymans1990specifying}, or \ac{STL} \cite{maler2004monitoring}, can be used to formulate different types of behavioral specifications. This has led to a wide array of control approaches involving guarantees of task satisfaction constraints through automata-based planning \cite{tabuada2006linear}, mixed-integer linear programming \cite{saha2016milp}, or feedback-based control laws \cite{lindemann2017prescribed}, among many others. Recently, reinforcement learning methods have also been investigated in this context \cite{sadigh2014learning, aksaray2016q, muniraj2018enforcing}.

Focusing on learning-based approaches, \ac{STL} is of special interest; unlike most other temporal logics, it allows definitions of robustness metrics associated to the degree of task satisfaction \cite{donze2010robust}. These metrics quantify how much a task is satisfied or violated, and serve as more descriptive rewards for learning than a simple true/false answer as to whether task satisfaction is achieved or not. Thus, robustness metrics can be seen as a general form of \textit{reward shaping} for the class of behaviors described by \ac{STL} specifications. Reward shaping is well-known to play a crucial role in the convergence of learning methods \cite{ng1999policy}, which motivates the study of robustness metrics from such a learning perspective.

Recently, many new extensions of the traditional robustness metric \cite{donze2010robust} for \ac{STL} have been introduced in the literature for various purposes. Examples include discretized and cumulative definitions aiming to ease computational burdens and to smoothen the robustness metric for use in real-time optimization \cite{lindemann2019robust, pant2017smooth, haghighi2019control}. These rely on point-wise or smooth approximations of the $\min$ and $\max$ operators in \cite{donze2010robust}, and do not preserve the desirable property of having the sign of the metric directly relate to the satisfaction of the corresponding \ac{STL} expression. In \cite{mehdipour2019arithmetic}, the authors define a metric that preserves this so-called \textit{soundness} property. The goal therein was to achieve higher robustness against noise by maximizing their metric instead of the traditional one.

To the extent of the authors' knowledge, this work is the first to consider \ac{STL} robustness metrics explicitly for their role in accelerating learning procedures. In particular, the focus is on the so-called \textit{shadowing problem} of the traditional robustness metric \cite{donze2010robust}: the definition of this metric is such that increasing the robustness of one term in a conjunction of propositions does not show in the robustness computed for the conjunction itself. This makes it more difficult for learning methods to find improvements towards task satisfaction through exploration. Our main contribution towards countering this problem is two-fold. First, we provide a theoretical study of some fundamental limitations involved in designing any robustness metric that is assumed to satisfy a set of chosen desirable properties for its role in aiding learning algorithms. Second, we use the findings to define a new class of robustness metrics specifically engineered for aiding exploration. The effectiveness of this new metric compared to previous ones is demonstrated in a case study.

The rest of the paper is organized as follows. Section II provides  background on \ac{STL}, its robustness metrics, and briefly outlines the algorithm used to compare their performance for aiding learning. Sections III and IV organize desirable properties for constructing metrics, and examine the restrictions these impose on them. Section V introduces a class of robustness metrics oriented towards accelerating learning, which is demonstrated by a simulation case study in Section VI. Concluding remarks are given in Section VII.

\section{PRELIMINARIES} \label{section:preliminaries}

\subsection{\Acf{STL}}

In \ac{STL}, the predicates are defined over continuous-time signals, such as the state $\x$ of the system \cite{maler2004monitoring}. The logical \textit{predicates} $\mu$ are either true ($\top$) or false ($\bot$), which is determined according to a corresponding function $h^{\mu}:\inReal*[n] \rightarrow \inReal*$ as $\mu = \top$ if $h^{\mu}(\x) \geq 0$ and $\mu = \bot$ otherwise. The predicates can be recursively combined using Boolean and temporal operators to form more complex \textit{task specifications} $\phi$:
\begin{equation*}
\phi := \top \ |\  \mu \ |\ \neg \phi \ |\ \phi_1 \and \phi_2 \ |\ \phi_1 \until[a][b]\phi_2,
\end{equation*}
where the \textit{until operator} $\until[a][b]$ requires $\phi_1$ to hold until $\phi_2$ eventually becomes true in the time interval defined by $[a, b]$. For a formal definition of the \ac{STL} language semantics, we refer to \cite{maler2004monitoring}. Briefly, satisfaction of an expression (denoted $(\x, t) \satisfies \phi$) can be computed recursively using the semantics $(\x, t) \satisfies \mu \iff h^{\mu}(\x(t)) \ge 0$; $(\x, t) \satisfies \neg\phi \iff \neg((\x, t) \satisfies \phi)$; $(\x, t) \satisfies \phi_1 \and \phi_2 \iff (\x, t) \satisfies \phi_1 \and (\x, t) \satisfies \phi_2$; and $(\x, t) \satisfies  \phi_1 \until[a][b]\phi_2 \iff \exists t_1 \in [t+a, t+b]$ such that $(\x, t_1) \satisfies \phi_2$ and $(\x, t_2) \satisfies \phi_1$ $\forall t_2 \in [t, t_1]$.  Other expressive temporal operators include \textit{eventually} and \textit{always}, whose definition follows from $\eventually[a][b]\phi = \top \until[a][b]\phi$ and $\always[a][b]\phi = \neg \eventually[a][b] \neg \phi$.
\subsection{Robustness metrics for \ac{STL}}

An advantage of \ac{STL} is that it allows the definition of different \textit{robustness metrics}, which give a quantitative indication of how well a task expression is satisfied. In this work, we consider spatial robustness metrics for \ac{STL}, i.e., metrics which give an indication towards how well the formula is satisfied given the imposed timing specifications. There are also notions of time robustness (e.g. \cite{akazaki2015time}), which quantify the extent to which the timing requirements are met.

The original metric as defined in \cite{donze2010robust} will be referred to as the \textit{traditional} robustness metric and is evaluated as follows for a selection of operators that are the focus of this work:
\begin{align}
\rho^\mu(\x, t) &= \rho^\mu(\x(t)) = h^{\mu}(\x(t)) \nonumber \\
\rho^{\neg \phi}(\x, t) &= -\rho^{\phi}(\x,t) \nonumber \\
\rho^{\phi_1 \and \phi_2}(\x, t) &= \min\left(\rho^{\phi_1}(\x, t),\rho^{\phi_2}(\x, t)\right) \label{eq:traditionalMetric} \\
\rho^{\eventually[a][b]\phi}(\x, t)  &= \max_{t' \in [t+a,t+b]}\rho^{\phi}(\x,t') \nonumber \\
\rho^{\always[a][b]\phi}(\x, t)  &= \min_{t' \in [t+a,t+b]}\rho^{\phi}(\x,t'). \nonumber
\end{align}

A basic property of this robustness metric is that its sign gives an explicit indication of whether or not its corresponding \ac{STL} specification is satisfied. Mathematically, this is described as the following property \cite{donze2010robust}.

\begin{property}[Soundness]
	A robustness metric is sound if, for any specification $\phi$, $\rho^\phi(\x, t) \ge 0$ if and only if the signal $\x$ satisfies $\phi$ at time $t$, i.e., $(\x, t) \satisfies \phi$.
\end{property}

Most robustness metrics recently defined in the literature \cite{lindemann2019robust, pant2017smooth, haghighi2019control} do not preserve this property in exchange for gains from a computational or optimization perspective. On the other hand, the metric \cite{mehdipour2019arithmetic} was defined using notions of arithmetic and geometric means in a way such that this property is preserved. The goal of this so-called AG metric was to better express robustness of some formulas, e.g., the conjunction of the set of metrics $\{1, 1, 1, 1, 1\}$ should receive a lower robustness score than the set $\{1, 10, 10, 10, 10\}$, whereas both achieve a score of $1$ using the traditional definition. Note that this behavior also helps counter the shadowing problem discussed in the previous section.

\subsection{Guided \acf{PI2}} \label{section:PI2}

\ac{PI2} is a reinforcement learning algorithm for solving optimal control problems in unknown environments. Following our recent work regarding its application to satisfying \ac{STL} tasks, its guided variant will be used to compare the performance of different robustness metrics in terms of accelerating the learning process. Due to space limitations, here we give a brief outline of the method, and refer to \cite{varnai2019prescribed} for details.

Guided \ac{PI2} searches for a parameterized policy $\pi(\x,t) = \u<\hat>(\x, t) + \k_\theta(t)$ in order to minimize a user-defined cost function $J(\tau)$ of the system trajectory $\tau$ from a given initial state $\x[0]$. Here, $\u<\hat>(\x, t)$ is a feedback controller used to guide exploration by aiming to impose given timing constraints, i.e., \textit{funnels} $\rho^{\mu_i}(\x(t)) \ge \gamma(t)$, on the evolution of the $i=1,\dots,M$ atomic propositions $\mu_i$ composing the \ac{STL} task $\phi$. The algorithm seeks to find the feedforward terms $\k_\theta(t)$ which (locally) minimize $J(\tau)$. In its $(k)$-th iteration, a set of $N$ parameters are sampled from around the current solution estimate $\theta^{(k)}$, the corresponding controllers' performances are evaluated using the cost $J(\tau)$, and the $\theta$ parameters are updated towards the more optimal ones. The iterations are repeated a given number of times or until convergence. In the context of satisfying $\phi$, the cost function $J(\tau)$ is composed of a cost of interest $C(\tau)$ to be minimized and a penalty term for progressively enforcing task satisfaction by penalizing negative values of the corresponding robustness $\rho^\phi(\x,0)$.

\section{STRUCTURED DEFINITION OF ROBUSTNESS METRICS} \label{section:robustnessDef}

In order to algorithmically calculate the robustness metrics of complex \ac{STL} formulas, it is useful to define them in a recursive manner. This was the case for the traditional robustness measure (\ref{eq:traditionalMetric}) given in the previous section; the expressions for the negation, conjunction, eventually, and always operators all rely on previously computed robustnesses. 

A further desirable property of any robustness metric is for its value to remain unaltered under invariant changes to \ac{STL} formulas themselves. For example, in terms of Boolean truth value, the expression $\mu_1 \and \mu_2$ is equivalent to $\neg (\neg \mu_1 \vee \neg \mu_2)$. Computing the robustness using the recursions defined by either expression should return the same value. Transferring various properties of \ac{STL} onto its robustness metrics thus allows us to establish identities between the operators of the latter, leading to its definition in a structured manner.

Formally, we introduce the abstract operators $\NEG$, $\AND$, $\OR$, $\EVENTUALLY$, and $\ALWAYS$ for negation, conjunction, disjunction, and the eventually and always operators. We then require the well-known De Morgan identities of Boolean and temporal logic to be transferred to the robustness metric itself. Thus,
\begin{equation*}
\OR(\phi_1, \dots, \phi_M) = \NEG(\AND(\NEG(\phi_1), \dots, \NEG(\phi_M)))
\end{equation*}
implies that the robustness metric must satisfy:
\begin{equation} \label{eq:orDef}
\OR(\rho_1, \dots, \rho_M) = \NEG(\AND(\NEG(\rho_1), \dots, \NEG(\rho_M))).
\end{equation}
Similarly, between the always and eventually operators, the identity
\begin{equation*}
\EVENTUALLY[a][b](\phi) = \NEG(\ALWAYS[a][b](\NEG(\phi)))
\end{equation*}
translates to
\begin{equation} \label{eq:eventuallyDef}
\EVENTUALLY[a][b](\rho^{\phi}(\x, t)) = \NEG(\ALWAYS[a][b](\NEG(\rho^\phi(\x,t)))).
\end{equation}
Further note that the always operator can be interpreted as a conjunction over a given time frame, such as in the definition of the traditional robustness metric (\ref{eq:traditionalMetric}). Discretizing the interval $[a, b]$ into $i=1,\dots,M$ evenly spaced points $t_i$, we can define 
\begin{equation} \label{eq:alwaysDef}
	\ALWAYS[a][b](\rho^{\phi}(\x, t)) = \lim_{M\rightarrow \infty}\AND(\rho^{\phi}(\x,t+t_1), \dots, \rho^{\phi}(\x,t+t_M))
\end{equation}

The identities (\ref{eq:orDef}), (\ref{eq:eventuallyDef}), and (\ref{eq:alwaysDef}) allow us to define new robustness metrics in a structured manner, from the elementary definitions of the $\NEG$ and $\AND$ operators. The former can also be excluded from the design by setting $\NEG(\rho^{\phi}) := -\rho^{\phi}$, a natural choice for quantifying the satisfaction or violation of a task symmetrically by the same magnitude robustness metric. Thus, the operators required to evaluate the robustness of \ac{STL} formulas recursively can be constructed given the definition of the single AND operator $\AND$, significantly easing the discussion related to constructing new metrics.

As an example, the arithmetic-geometric mean robustness metric of \cite{mehdipour2019arithmetic} (denoted by `AG') can be constructed from:
\begin{equation} \label{eq:AGand}
\AND^{\text{AG}}(\rho_1,\dots,\rho_M)
= \begin{cases}
\frac{1}{M}\sum_{i=1}^{M} \max(\rho_i, 0)\ \ \text{if } \min_i \rho_i \le 0,\\
\sqrt[M]{\prod_{i=1}^{M}(1 + \rho_i)}-1 \ \ \text{otherwise.}
\end{cases}
\end{equation}
Similarly, the traditional robustness metric (\ref{eq:traditionalMetric}) can be constructed from $\AND^{\text{trad}}(\rho_1,\dots,\rho_M) = \min(\rho_1,\dots,\rho_M)$.

\begin{remark}
	For evaluating the always operator using (\ref{eq:alwaysDef}), summations and products in the definition of $\AND$ are replaced by integrals and product integrals in the limit $M \rightarrow \infty$: considering a time horizon $[a, b]$ with a discretized set of points $t_1, \dots, t_M \in [a, b]$, we have $\lim_{M \rightarrow \infty} \frac{1}{M} \sum_{i=1}^{M} f(t_i) = \frac{1}{b-a}\int_{a}^{b}f(t) \diff t$, and for products $\lim_{M \rightarrow \infty} \sqrt[M]{\prod_{i=1}^{M} f(t_i)} = e^{\frac{1}{b-a}\int_{a}^{b} \ln f(t) \diff t}$. Such sums and products appear in the AG metric (\ref{eq:AGand}) and the proposed metric (\ref{eq:defSNC}) as well, though we keep the discretized form of the definitions for simplicity.
\end{remark}


\section{PROPERTIES OF ROBUSTNESS METRICS} \label{section:solution}

Based on the previous section, it is convenient to define robustness metrics through their AND operator $\AND$. In the following, we examine desirable properties of this operator from a learning-based perspective, as well as some fundamental restrictions that their satisfaction imposes on the form of the operator itself. In order to be clearer with notation, we denote the operator for a conjunction of $M$ terms by $\AND[M]$.

\subsection{Desirable properties}

We consider desirable properties of the AND operator from a learning-based perspective, i.e., for the potential role of robustness metrics as rewards for learning-based methods. 

1) First, consider the fundamental Boolean identities of idempotence and commutativity:
\begin{property}[Idempotence, commutativity]
	The AND operator $\AND[M]$ is idempotent and commutative if:
	\begin{enumerate}[(i)]
		\item $\AND[M](\rho,\dots,\rho) = \rho$, and
		\item for any permutation ${k_i}$ of the integers $i = 1, \dots, M$, $\AND[M](\rho_i,\dots,\rho_M) = \AND[M](\rho_{k_i},\dots,\rho_{k_M})$.
	\end{enumerate}
\end{property}
Note that associativity is not a fundamental property we wish to preserve, as it contradicts a later, more desirable property. This is also the reason we define the $\AND[M]$ operator as a function of $M$ operands; e.g., the identity $\AND[3]\left(\rho_1, \rho_2, \rho_3\right) = \AND[2]\left[\rho_1, \AND[2](\rho_2, \rho_3)\right]$, whose equivalent holds for Boolean logic, will not hold in general.

2) Second, from a general optimization perspective, it is desirable for the operator $\AND$ to be smooth in order to aid gradient-based and acceleration methods.

\begin{property}[Weak smoothness]
	$\AND[M](\rho_1,\dots,\rho_M)$ is \textit{weakly smooth} if (i) it is continuous everywhere, and (ii) if its gradient is continuous for all points for which there are no two indices $i \ne j$ satisfying $\rho_i = \rho_j = \min_{k \in \{1,\dots,M\}} \rho_k$.
\end{property}
Unlike simply requiring smoothness almost everywhere, this definition expresses the desire for smoothness at points where there is a unique minimal term. In particular, a special point of interest is when the robustness metric switches sign, indicating that the corresponding \ac{STL} specification has become true or false.  Neither the traditional nor the AG robustness metrics are smooth at such points.

3) Third, from a learning-based perspective of guiding towards task satisfaction, it is important to address the \textit{shadowing problem} outlined in the introduction. With the traditional metric $\AND[M](\rho_1, \dots, \rho_M) = \min(\rho_1, \dots, \rho_M)$, an increase of any $\rho_i$ is not seen in the robustness of the conjunction, unless $\rho_i$ was the unique minimum of the $M$ terms. For example, if each $\rho_i = \rho$, this behavior is undesired; the reward should indicate that an increase of any $\rho_i$ would at some point be beneficial for increasing the robustness of the conjunction itself. The shadowing problem thus makes it difficult for a learning algorithm to find improvements towards task satisfaction. The mathematical formulation of tackling the shadowing problem is given as:
\begin{property}[Shadow-lifting property] \label{property:shadowingProperty}
	The operator $\AND[M]$ satisfies the shadow-lifting property if, for any $\rho \ne 0$, $\left.\frac{\partial \AND[M]\left(\rho_1, \dots, \rho_i, \dots, \rho_M \right)}{\partial \rho_i}\right|_{\rho_1, \dots, \rho_M = \rho} > 0$ holds $ \forall i = 1, \dots, M$.
\end{property}
Considering a set of points $\rho_1, \dots, \rho_M = \rho$, the robustness metric of their conjunction using the traditional metric would only show an increase if all $\rho_i$ terms increase. On the other hand, the shadow-lifting property implies that $\AND[M]\left(\rho, \dots, \rho \right)$ also increases when making partial progress towards this goal and increasing only a set of the $\rho_i$ terms. The more elements of the conjunction change, the greater increase we see in the robustness due to the linearity of the partial derivatives. The AG robustness satisfies the shadow-lifting property; however, we will see that being too rewarding for positive changes of the terms in a conjunction has pitfalls related to local minima in more complex tasks.
\begin{remark}
	The shadow-lifting property prohibits the associative property of the AND operator to be satisfied. For example, if it were, the robustness of the conjunction of the two sets of robustness metrics $\{1, 1, 1+\epsilon\}$ and $\{1, 1+\epsilon, 1+\epsilon\}$ would be both equivalent to that of $\{1, 1+\epsilon\}$, whereas the second one should be higher.
\end{remark}

4) Finally, we consider two additional unclassified properties that impose natural restrictions on the AND operator. 
\begin{property}[min/max boundedness] \vspace{-1mm}
	The operator $\AND[M]$ is min/max bounded if it satisfies the inequality \vspace{-1mm}
	\begin{equation*} 
		\min(\rho_1,\dots,\rho_M) \le \AND[M](\rho_1,\dots,\rho_M) \le \max(\rho_1,\dots,\rho_M).
	\end{equation*}
\end{property}
\begin{property}[Scale-invariance]
	The operator $\AND[M]$ is said to be scale-invariant if, for any $\alpha \ge 0$, it satisfies the identity\vspace{-1mm}
	\begin{equation} \label{eq:scaleInvariance} \vspace{-1mm}
	\AND[M](\alpha\rho_1, \dots, \alpha\rho_M) = \alpha\AND[M](\rho_1, \dots, \rho_M).
	\end{equation}
\end{property}

The former property is useful for placing fundamental restrictions on the values of $\AND[M](\cdot)$. The latter is desirable for the AND function to behave similarly regardless of the order of magnitude of its robustness metric terms, e.g, in case we do not know their order of magnitude in advance. 

Table I on the right summarizes the properties satisfied by the traditional and AG robustness metrics.

\subsection{Imposed restrictions}

Properties 1-6 impose fundamental restrictions on the form of the operator $\AND[M]$ used for constructing any sought-after robustness metric. In the following, various propositions regarding these restrictions are presented in order to motivate the definition of a new robustness metric in Section \ref{section:newMetric}.

\begin{proposition}
	The operator $\AND[M]$ cannot be sound, idempotent, and smooth simultaneously.
\end{proposition}
\begin{proof}
	Consider the behavior of $\AND[2](\rho_1, \rho_2)$ for $\rho_1 = 0$ and $\rho_2 > 0$.  By soundness, $\AND[2](\rho_1,\rho_2)$ must switch sign as $\rho_1$ switches sign; continuity therefore implies $\AND[2](0,\rho_2) = 0$ for any $\rho_2 > 0$. In particular, $\AND[2](0,\rho_2)$ remains 0 when $\rho_2 \rightarrow 0+$, thus $\left.\frac{\partial \AND[2](\rho_1, \rho_2)}{\partial \rho_2}\right|_{\rho_1,\rho_2 = 0} = 0$. The same argument holds for the partial derivative with respect to $\rho_1$ at this point. For $\epsilon \rightarrow 0+$, in the first-order approximation we must thus have $\AND[2](\epsilon,\epsilon) = \left(\frac{\partial \AND[2](\rho_1,\rho_2)}{\partial \rho_1} + \frac{\partial \AND[2](\rho_1,\rho_2)}{\partial \rho_2}\right)\epsilon = 0$. This, however, contradicts $\AND[2](\epsilon, \epsilon) = \epsilon$ implied by idempotence.
\end{proof}

The proposition implies that smoothness across the entire domain is a too strict requirement. As we will see, however, the operator $\AND[M]$ can be \textit{weakly smooth}. An interesting question is if the gradient could be continuous at more points than required by weak smoothness; e.g., one could require continuity of the gradient at all points except where there are two indices $i,j$ such that $\rho_i = \rho_j = \min_k \rho_k \ne 0$. Note the added `$\ne 0$' condition compared to weak smoothness. Although the study of this question is outside the scope of this work, the metric defined in the next section will actually be smooth at points where all $\rho_i$ are equal, giving merit to the idea. The following proposition relates smoothness at such a key point of interest to the shadow-lifting property.

\begin{proposition}
	Assume the AND operator $\AND[M]$ is defined such that it is smooth and satisfies Property 2. Then, for any $\rho \ne 0$, the operator satisfies
	\begin{equation} \label{eq:equalityPartial}
	\left.\frac{\partial \AND[M](\rho_1, \dots, \rho_M)}{\partial \rho_i}\right|_{\rho_1, \dots, \rho_M = \rho} = \frac{1}{M}, \ \forall i=1,\dots,M
	\end{equation}
	which is positive and thus implies that Property 4 also holds.
\end{proposition}
\begin{proof}
	If the gradient is continuous, then for any $\rho \ne 0$ in the first-order approximation we must have:
	\begin{align*}
		\lim_{\epsilon \rightarrow 0} \AND[M](\rho + \epsilon, \dots, &\rho + \epsilon) = \AND[M](\rho, \dots, \rho) \\ &+ \sum\nolimits_{i=1}^{M}\epsilon \left.\frac{\partial \AND[M](\rho_1, \dots, \rho_M)}{\partial \rho_i}\right|_{\rho_1,\dots,\rho_M = \rho}
	\end{align*}
	The idempotence property then implies the equality:
	\begin{equation*}
		\rho + \epsilon = \rho + \epsilon\sum\nolimits_{i=1}^{M}\left.\dfrac{\partial \AND[M](\rho_1, \dots, \rho_M)}{\partial \rho_i}\right|_{\rho_1,\dots,\rho_M = \rho}
	\end{equation*}
	must hold. Furthermore, due to commutativity of the variables, the partial derivatives at $\rho_1 = \ldots = \rho_M$ must equal one another, which implies the desired result (\ref{eq:equalityPartial}).
\end{proof}

\begin{table}[t]
	\centering 
	\caption{Summary of the discussed properties satisfied by the two robustness metrics introduced in the literature.}
	\label{tab:satisfactions}
	\resizebox{0.78\linewidth}{!}{%
		\begin{tabular}{|l|c|c|c|c|c|c|}
			\hline
			\multicolumn{1}{|c|}{\multirow{2}{*}{Robustness metric}} & \multicolumn{6}{c|}{Property number} \\ \cline{2-7} 
			\multicolumn{1}{|c|}{}                                   & 1     & 2    & 3   & 4    & 5    & 6    \\ \hline
			traditional \cite{donze2010robust}                       & \cmark   & \cmark  & \xmark  & \xmark   & \cmark  & \cmark  \\ \hline
			AG \cite{mehdipour2019arithmetic}                        & \cmark   & \cmark  & \xmark  & \cmark  & \cmark  & \xmark   \\ \hline
		\end{tabular}%
	}
	\vspace{-3mm}
\end{table}

Next, we consider restrictions that weak smoothness imposes on the gradient at points where the minimal term in the conjunction is unique and equal to zero.

\begin{proposition} \label{theorem:partialD0}
	Assume the AND operator $\AND[M]$ is defined such that Properties 1-3 and 5 hold. Then,
	\begin{equation} \label{eq:partialD0}
	\left.\dfrac{\partial \AND[M](\rho_1, \dots, \rho_M)}{\partial \rho_i}\right|_{\rho_i = 0,\ \rho_{j \ne i} > 0} = 1, \ \forall i=1,\dots,M.
	\end{equation}
\end{proposition}
\begin{proof}
	Consider the case $i=1$ without loss of generality due to the commutative property of $\AND[M]$. As $\AND[M]$ is sound, it must switch from positive to negative as $\rho_1$ switches from positive to negative at points where $\rho_1 = 0, \rho_{j \ne 1} > 0$. Therefore, as $\AND[M]$ is continuous, we must have $\AND[M](\rho_1, \dots, \rho_M)$ = 0 at such points. The partial derivative $\left.\frac{\partial \AND[M](\rho_1, \dots, \rho_M)}{\partial \rho_i}\right|_{\rho_i = 0,\ \rho_{j \ne i} > 0}$ can thus be evaluated as:
	\begin{equation} \label{eq:gradLimitDef}
	\lim_{\rho_1 \rightarrow 0} \dfrac{\AND[M](\rho_1, \dots, \rho_M) - 0}{\rho_1 - 0} = \lim_{\rho_1 \rightarrow 0} \dfrac{\AND[M](\rho_1, \dots, \rho_M)}{\rho_1}
	\end{equation}
	For the limit from below, since $\rho_1$ is the minimal term, the min/max bound inequality $\rho_1 \le \AND[M](\rho_1, \dots, \rho_M) (< 0)$ implies $\lim_{\rho_1 \rightarrow 0-} \frac{\AND[M](\rho_1, \dots, \rho_M)}{\rho_1} \le 1$. On the other hand, taking the limit from above, the inequality $(0 <) \rho_1 \le \AND[M](\rho_1, \dots, \rho_M)$ implies $\lim_{\rho_1 \rightarrow 0+} \frac{\AND[M](\rho_1, \dots, \rho_M)}{\rho_1} \ge 1$. For weak smoothness, the two limits must match, thus the partial derivative (\ref{eq:partialD0}) must be equal to 1 at this point.
\end{proof}

Finally, the following proposition relates possible values $\AND[M]$ can take in case scale-invariance is also assumed.

\begin{proposition} \label{theorem:inftyBehavior}
	Assume the AND operator $\AND[M]$ is defined such that Properties 1-3, 5, and 6 hold, and consider without loss of generality $\rho_1 \ne 0$ fixed. Then, the condition (\ref{eq:partialD0}) on the partial derivative is satisfied if and only if:
	\begin{equation} \label{eq:partialDinf}
		\lim\limits_{\rho_2,\dots,\rho_M \rightarrow \infty} \AND[M]\left(\rho_1, \dots, \rho_M\right) = \rho_1.
	\end{equation}
\end{proposition}

\begin{proof}
	The scale-invariance property allows us to relate the form of limits given by equations (\ref{eq:gradLimitDef}) and (\ref{eq:partialDinf}). Let $\epsilon$ have the same sign as $\rho_1 \ne 0$, and consider a set of values $\bar{\rho}_{i \ne 1} > 0$. Using (\ref{eq:scaleInvariance}) with $\alpha := \epsilon / \rho_1 > 0$, one obtains the relation:
	\begin{equation} \label{eq:scaling}
	\dfrac{\rho_1}{\epsilon} \AND[M](\epsilon, \bar{\rho}_2, \dots, \bar{\rho}_M) = \AND[M](\rho_1, \bar{\rho}_2\dfrac{\rho_1}{\epsilon}, \dots, \bar{\rho_M}\dfrac{\rho_1}{\epsilon}).
	\end{equation}
	For the `if' direction, assuming (\ref{eq:partialDinf}) holds, then for $\epsilon \rightarrow 0$ the right hand side of this equation becomes $\rho_1$. Dividing both sides by $\rho_1$ then yields:
	\begin{equation} \label{eq:rho1Limit}
		\lim\limits_{\epsilon \rightarrow 0} \dfrac{1}{\epsilon} \AND[M](\epsilon, \bar{\rho}_2, \dots, \bar{\rho}_M) = 1,
	\end{equation}
	which by definition of the partial derivative implies (\ref{eq:partialD0}). Conversely, for the `only if` direction, assuming (\ref{eq:partialD0}) holds, by definition (\ref{eq:rho1Limit}) must hold as well. Substituting this into (\ref{eq:scaling}) for $\epsilon \rightarrow 0$ shows that now (\ref{eq:partialDinf}) holds, as desired.
\end{proof}

\begin{remark}
Proposition \ref{theorem:inftyBehavior} and the shadow-lifting property imply that any AND operator satisfying Properties 1-6 altogether is non-monotone. This can be seen in the simple case of considering $\AND[2](-1, \rho)$ as $\rho$ increases from -1 towards infinity. Initially, the robustness of the conjunction must increase, but drop back to the minimal element -1 as $\rho \rightarrow \infty$.
\end{remark}

\section{NEW ROBUSTNESS METRIC} \label{section:newMetric}

We introduce a family of robustness metrics following the imposed restrictions uncovered in the previous section. More specifically, we construct an AND operator that satisfies all of Properties 1-6. The construction is such that a parameter $\nu > 0$ controls how closely the defined metric approaches the traditional $\AND^{\text{trad}}(\rho_1, \dots, \rho_M) = \min_i \rho_i := \rho_{\min}$ operator.

\subsection{Behavior for $\rho_{\min} < 0$}

If $\rho_{\min} < 0$, the conjunction is not satisfied, hence for soundness and continuity $\AND[M](\rho_1,\dots,\rho_M)$ must remain negative and approach $0$ as $\rho_{\min} \rightarrow 0-$. Aiming at a scale-invariant behavior, let us define:
\begin{equation} \label{eq:normalized_rho}
\tilde{\rho}_i := \dfrac{\rho_i - \rho_{\min}}{\rho_{\min}},
\end{equation}
which is non-positive (as $\rho_{\min} < 0$) and becomes $0$ at $\rho_i = \rho_{\min}$. Using this normalized measure, we further define the effective robustness measures:
\begin{equation} \label{eq:effectiveDef}
\rho^{\text{eff}}_i := \rho_{\min} e^{\tilde{\rho}_i},
\end{equation}
whose purpose is to transform each $\rho_i$ such that $\rho^{\text{eff}}_i$ is negative and remains between $\rho_{\min}$ and $\rho_i$. The AND operator for $\rho_{\min} < 0$ with parameter $\nu > 0$ is then defined by taking the weighted average of these effective measures:
\begin{equation} \label{eq:defNeg}
\AND[M]^{\text{new}}(\rho_1, \dots, \rho_M) := \dfrac{\sum_{i=1}^M\rho^{\text{eff}}_i e^{\nu \tilde{\rho}_i}}{\sum_{i=1}^M e^{\nu \tilde{\rho}_i}}, \ \text{if } \rho_{\min} < 0.
\end{equation}
The weighting function $e^{\nu \tilde{\rho}_i}$ is such that the weight is $1$ for $\rho_i = \rho_{\min}$ and becomes $0$ as $\rho_i \rightarrow \infty$; this property is motivated by the desire to satisfy the condition $\lim\limits_{\rho_i \ne \rho_{\min} \rightarrow \infty} \AND[M]\left(\rho_1, \dots, \rho_M\right) = \rho_{\min}$ imposed by Proposition \ref{theorem:inftyBehavior}. Also note that as $\nu \rightarrow \infty$, the weights $e^{\nu \tilde{\rho}_i}$ tend to 0 for $\rho_i \ne \rho_{\min}$ because then $\tilde{\rho}_i < 0$, and the defined operator becomes the traditional $\AND[M]^{\text{trad}} = \min_i \rho_i$ metric.

\subsection{Behavior for $\rho_{\min} > 0$}

If $\rho_{\min} > 0$, the conjunction is satisfied, hence for soundness and continuity $\AND[M](\rho_1, \dots, \rho_M)$  must remain positive and approach $0$ as $\rho_{\min} \rightarrow 0+$. For a structurally scale-invariant behavior, we again use the normalized measure defined by (\ref{eq:normalized_rho}). As opposed to the previous case, we do not require an effective measure in order to impose negative values on the metric, and can readily take the weighted average of the robustnesses forming the conjunction as:
\begin{equation} \label{eq:defPos}
\AND[M]^{\text{new}}(\rho_1, \dots, \rho_M):= \dfrac{\sum_i \rho_i e^{-\nu \tilde{\rho}_i}}{\sum_i e^{-\nu \tilde{\rho}_i}}, \ \text{if } \rho_{\min} > 0.
\end{equation}
Note that if $\rho_{\min} > 0$, the normalized measures are all positive, and the exponential weights become $1$ for $\rho_i = \rho_{\min}$ and $0$ as $\rho_i \rightarrow \infty$, as in the previous case. As before, setting $\nu \rightarrow \infty$ also reduces (\ref{eq:defPos}) to the traditional $\min$ operator.

\subsection{Definition}
The new robustness measure is defined from (\ref{eq:defNeg}) and (\ref{eq:defPos}) as:
\begin{equation}\label{eq:defSNC}
	\AND[M]^{\text{new}}(\rho_1, \dots, \rho_M) = \begin{cases}
		\dfrac{\sum_i \rho_{\min} e^{\tilde{\rho}_i} e^{\nu \tilde{\rho}_i}}{\sum_i e^{\nu \tilde{\rho}_i}} \quad \text{if } \rho_{\min}  < 0, \\
		\dfrac{\sum_i \rho_i e^{-\nu \tilde{\rho}_i}}{\sum_i e^{-\nu \tilde{\rho}_i}} \qquad \ \ \text{if } \rho_{\min}  > 0, \\
		0 \qquad \qquad \qquad \ \ \ \, \text{if }\rho_{\min} = 0.
	\end{cases}
\end{equation}
Sample behaviors of the operator $\AND[2]^{\text{new}}(\rho_1, \rho_2)$ for different configurations of the variables $\rho_1$ and $\rho_2$ are depicted in Figure \ref{fig:newcurves}. Note that the illustrated curves are smooth, unlike in case of the traditional and AG robustness measures. However, monotonicity is not achieved, as expected from Theorem \ref{theorem:inftyBehavior}.

\begin{figure}[hb]
	\centering
	\begin{tikzpicture} \hspace{-6mm}
	\captionsetup{format=plain,labelformat=parens,labelsep=space}
	\setcounter{subfigure}{0}
	\begin{groupplot}[group style={group name=myplots,group size= 2 by 1,horizontal sep=6mm},width=0.5\linewidth]
	\nextgroupplot[width=0.6\linewidth,grid=both,xmin=-1.1,ymin=-1.1,ymax=-0.42,xmax=0.1,xlabel=$\rho_2$,ylabel={$\AND(\rho_1,\rho_2)$},axis x line=bottom, axis y line = left,ylabel near ticks,ytick = {-1,-0.5},ylabel style={font=\small},ylabel shift = -15pt,minor x tick num=1, minor y tick num = 4,grid style={line width=.2pt, draw=gray!50}, major grid style={line width=.4pt,draw=gray!80},yscale=1.6,xscale=1.0,tick label style={font=\scriptsize}]
	\addplot[domain=-1.1:0.1, samples=100, lightblue, very thick] (x,{min(x,-1)}); \label{plots:min}
	\addplot[domain=-1.1:0.1, samples=100, lightred, very thick] (x,{(min(x,0)-1)/2}); \label{plots:AG}
	\addplot[domain=-1.1:-1.0, samples=10, cyan, very thick] (x,{x*(1 + exp(2*(-1-x)/x))/(1 + exp(1*(-1-x)/x))}); \label{plots:nu1}
	\addplot[domain=-1.0:0.1, samples=80, cyan, very thick] (x,{-(1 + exp(-2*(x+1)))/(1 + exp(-1*(1+x)))});
	\addplot[domain=-1.1:-1.0, samples=10, darkgreen, very thick] (x,{x*(1 + exp(4*(-1-x)/x))/(1 + exp(3*(-1-x)/x))}); \label{plots:nu3}
	\addplot[domain=-1.0:0.1, samples=80, darkgreen, very thick] (x,{-(1 + exp(-4*(x+1)))/(1 + exp(-3*(1+x)))});
	
	\nextgroupplot[width=0.6\linewidth,grid=both,xmin=-0.2,ymin=-0.2,ymax=1.16,xmax=1.2,xlabel=$\rho_2$,axis x line=bottom, axis y line = left, minor x tick num=1, minor y tick num = 4,ytick = {0,1.0},grid style={line width=.2pt, draw=gray!50}, major grid style={line width=.4pt,draw=gray!80},yscale=1.6,xscale=1.0,tick label style={font=\scriptsize}]
	\addplot[domain=-0.2:1.2, samples=100, lightblue, very thick] (x,{min(x,1)});
	\addplot[domain=0.0:0.42, samples=20, lightred, very thick, dotted] (0,x);
	\addplot[domain=-0.2:0.0, samples=100, lightred, very thick] (x,{x/2});
	\addplot[domain=0.0:1.2, samples=100, lightred, very thick] (x,{sqrt(2*(1+x))-1});
	\addplot[domain=-0.2:-0.01, samples=10, cyan, very thick] (x,{x*(1 + exp(2*(1-x)/x))/(1 + exp(1*(1-x)/x))});
	\addplot[domain=-0.01:0.01, samples=2, cyan, very thick] (x,x);
	\addplot[domain=0.01:1.0, samples=90, cyan, very thick] (x,{(x + exp(-1*(1-x)/x))/(1 + exp(-1*(1-x)/x))});
	\addplot[domain=1.0:1.2, samples=80, cyan, very thick] (x,{(1 + x*exp(-1*(x-1)))/(1 + exp(-1*(x-1)))});
	\addplot[domain=-0.2:-0.01, samples=10, darkgreen, very thick] (x,{x*(1 + exp(4*(1-x)/x))/(1 + exp(3*(1-x)/x))});
	\addplot[domain=-0.01:0.01, samples=2, darkgreen, very thick] (x,x);
	\addplot[domain=0.01:1.0, samples=90, darkgreen, very thick] (x,{(x + exp(-3*(1-x)/x))/(1 + exp(-3*(1-x)/x))});
	\addplot[domain=1.0:1.2, samples=80, darkgreen, very thick] (x,{(1 + x*exp(-3*(x-1)))/(1 + exp(-3*(x-1)))});
	\end{groupplot}
	\node[text width=6cm,align=center,anchor=north] at ([yshift=-8mm]myplots c1r1.south) {\captionof{subfigure}{$\rho_1 = -1$ kept fixed}};
	\node[text width=6cm,align=center,anchor=north] at ([yshift=-8mm]myplots c2r1.south) {\captionof{subfigure}{$\rho_1 = 1$ kept fixed}};
	\path (myplots c1r1.north west|-current bounding box.north)--
	coordinate(legendpos)
	(myplots c2r1.north east|-current bounding box.north);
	\matrix[
	matrix of nodes,
	anchor=south,
	draw,
	inner sep=0.1em,
	draw
	]at([yshift=1ex, xshift=-2mm]legendpos)
	{
		\ref{plots:min}& \scriptsize traditional &[3pt]
		\ref{plots:AG}& \scriptsize AG &[3pt]
		\ref{plots:nu1}& \scriptsize new ($\nu = 1$) &[3pt]
		\ref{plots:nu3}& \scriptsize new ($\nu = 3$) \\};
	\end{tikzpicture}
	\caption{Robustness metric of the conjunction $\AND[2](\rho_1, \rho_2)$ for the traditional, AG, and newly proposed metrics as a function of $\rho_2$ for fixed values of $\rho_1$. On the left, the shadow-lifting property is illustrated as the robustness increases even though the minimum of the two terms is constant. Note that the curves corresponding to the proposed metric are smooth.}
	\label{fig:newcurves}
\end{figure}

The introduced robustness metric satisfies each of the desired Properties 1-6, as shown in the following theorem.

\begin{theorem}
	The AND operator defined by (\ref{eq:defSNC}) satisfies all of Properties 1-6.
\end{theorem}
\begin{proof}
	Due to space constraints, we prove the theorem for the definition of $\AND[M]^{\text{new}}$ in case $\min_i \rho_i  < 0$; the proof for $\min_i \rho_i > 0$ follows the same pattern.
	
	\textit{Properties 1-2:} For soundness, we must have $\AND[M]^{\text{new}}(\cdot) < 0$, which follows from $\rho_{\min} = \min_i \rho_i  < 0$ and the positiveness of all exponential factors in the definition (\ref{eq:defNeg}) of the operator. For idempotence, substituting in $\rho_1 = \dots \rho_M := \rho$ implies $\rho_{\min} = \rho$, and for all normalized metrics we have $\tilde{\rho}_i = 0$. The expression (\ref{eq:defNeg}) therefore reduces to $\sum_{i=1}^{M} \rho / \sum_{i=1}^{M} 1 = \rho$, as desired. Commutativity follows from the commutativity of $\rho_{\min} = \min_i \rho_i$ and from commutativity of addition.

	\textit{Property 3:} For weak smoothness, we first show that the operator is continuous. This is clearly the case due to the continuity of the composing $\min$ and exponential functions when $\rho_{\min} < 0$; furthermore, we must show $\AND[M]^{\text{new}}(\rho_1,\dots,\rho_M) \rightarrow 0$ as $\rho_{\min} \rightarrow 0-$, since by definition (\ref{eq:defSNC}) $\AND[M]^{\text{new}}(\rho_1,\dots,\rho_M) = 0$ when $\rho_{\min} = 0$. To show this, note that the effective robustness measures (\ref{eq:effectiveDef}) are defined such that $\rho_{\min} \le \rho^{\text{eff}}_i < 0$ $\forall i$, hence the weighted average (\ref{eq:defNeg}) of them will also satisfy these bounds. Thus, when $\rho_{\min} \rightarrow 0-$, $\AND[M]^{\text{new}}(\rho_1,\dots,\rho_M) \rightarrow 0$, as desired.
	
	It is clear that the gradient of (\ref{eq:defNeg}) is smooth whenever there is a unique index $i$ such that $\rho_i = \rho_{\min} < 0$ due to the smoothness of the composing functions. To complete the proof, we need to show that as $\rho_{\min} \rightarrow 0$, the left and right partial derivatives with respect to each $\rho_i$ become equal. Proposition \ref{theorem:partialD0} implies that the partial derivative with respect to the minimal term needs to be equal to 1, which will now be shown for the left side derivative. Indeed, let $i$ be the unique index for which $\rho_i = \min_i \rho_i$; then as $\rho_{\min}  \rightarrow 0-$, all other $\rho_{j \ne i} > 0$ and thus $\tilde{\rho}_{j \ne i} \rightarrow -\infty$ while $\tilde{\rho}_i = 0$. Substituted into the expression (\ref{eq:defNeg}), $\AND[M]^{\text{new}}(\rho_1,\dots,\rho_M)$ thus reduces to $\rho_{\min} = \rho_i$, and the partial derivative is::
	\begin{align*}
		\lim\limits_{\rho_i \rightarrow 0, \rho_{j \ne i} > 0} &\dfrac{\AND[M]^{\text{new}}(\rho_1,\dots,\rho_M) - \AND[M]^{\text{new}}(0,\rho_2,\dots,\rho_M) }{\rho_i - 0} = \\ &= \dfrac{\rho_i - 0}{\rho_i} = 1,
	\end{align*}
	as desired. Furthermore, the partial derivative with respect to the other variables must approach 0 as $\rho_{\min} \rightarrow 0-$, because $\AND[M](\cdot) \equiv 0$ at $\rho_{\min} = 0$. This can be shown by deriving the expression for these partial derivatives while treating $\rho_i = \rho_{\min}$ as a constant. Indeed, with respect to a variable $\rho_j$, $j \ne i$, one obtains:
	\begin{equation}
	\dfrac{\partial \AND[M](\cdot)}{\partial \rho_j} = e^{(1+\nu)\tilde{\rho}_j}\dfrac{(1 + \nu)(1 + \sum_{k \ne i} e^{\nu \tilde{\rho}_k}) - \nu e^{\nu\tilde{\rho}_j}}{\left(1 + \sum_{k \ne i} e^{\nu \tilde{\rho}_k}\right)^2}
	\end{equation}
	In the limit $\rho_{\min} \rightarrow 0-$, all $\tilde{\rho}_k \rightarrow -\infty$ for any $k \ne i$, and thus the above expression becomes 0. This completes the proof of weak smoothness.
	
	\textit{Property 4:} For the shadow-lifting property, we show that when $\rho_1,\dots,\rho_M := \rho \ne 0$, the partial derivative with respect to any $\rho_i$ becomes $1/M > 0$. Due to commutativity, without loss of generality consider limits as $\rho_1 \rightarrow \rho$, first from above with $\rho_1 = \rho + \epsilon$ and $\epsilon \rightarrow 0+$. Then $\rho_{\min} = \rho$ and each $\tilde{\rho}_{j \ne 1} = 0$, thus:
	\begin{align*}
		\lim_{\epsilon \rightarrow 0+} &\dfrac{1}{\epsilon}\left[\AND[M](\rho + \epsilon,\rho,\dots,\rho) -  \AND[M](\rho,\dots,\rho)\right] \\ &= \lim_{\epsilon \rightarrow 0+}\dfrac{\rho}{\epsilon} \dfrac{(M-1) + e^{(1+\nu)\frac{\rho + \epsilon- \rho}{\rho}}}{(M-1) + e^{\nu\frac{\rho + \epsilon - \rho}{\rho}}} - \dfrac{\rho}{\epsilon} \\
		&= \lim_{\epsilon \rightarrow 0+}\dfrac{\rho}{\epsilon} \dfrac{e^{(1+\nu)\frac{\epsilon}{\rho}} - e^{\nu\frac{\epsilon}{\rho}}}{(M-1) + e^{\nu\frac{\epsilon}{\rho}}} \\
		&= \lim_{\epsilon \rightarrow 0+} \dfrac{1}{(M-1)+e^{\nu\frac{\epsilon}{\rho}}} \lim_{\epsilon \rightarrow 0+} e^{\nu\frac{\epsilon}{\rho}} \lim_{\epsilon \rightarrow 0+} \dfrac{e^\frac{\epsilon}{\rho} - 1}{\epsilon / \rho} \\
		&= \tfrac{1}{M} \cdot 1 \cdot 1 = \tfrac{1}{M}.
	\end{align*} 
	For the limit from below, with $\rho_1 = \rho - \epsilon$ and $\epsilon \rightarrow 0+$, we have $\rho_{\min} = \rho - \epsilon$ and thus following the definition (\ref{eq:defNeg}):
	 \begin{align*}
		 \lim_{\epsilon \rightarrow 0+} &\dfrac{1}{\epsilon}\left[\AND[M](\rho - \epsilon,\rho,\dots,\rho) -  \AND[M](\rho,\dots,\rho)\right] \\ &= \lim_{\epsilon \rightarrow 0+}\dfrac{\rho - \epsilon}{\epsilon} \dfrac{1 + (M-1)e^{(1+\nu)\frac{\rho - (\rho - \epsilon)}{\rho - \epsilon}}}{1 + (M-1)e^{\nu\frac{\rho - (\rho - \epsilon)}{\rho - \epsilon}}} - \dfrac{\rho}{\epsilon}. 
	 \end{align*} 
	 Following the same steps as previously, this limit can also be shown to be $\frac{1}{M}$. Since the two limits equal, the partial derivative exists and is equal to $1/M$, as desired.
	 
	 \textit{Property 5:} For the boundedness of $\AND[M](\rho_1, \dots, \rho_M)$ by $\min_i \rho_i$ and $\max_i \rho_i$, note that (\ref{eq:defNeg}) takes the weighted average of the terms $\rho^{\text{eff}}_i = \rho_{\min} e^{\frac{\rho_i - \rho_{\min}}{\rho_{\min}}}$. As all $\rho_i > \rho_{\min}$ and $\rho_{\min} < 0$, the exponent is negative and thus each $\rho^{\text{eff}}_i \ge \rho_{\min}$. Furthermore, by the Bernoulli inequality $e^x \ge 1 + x$:
	 \begin{equation*}
	 \rho^{\text{eff}}_i = \rho_{\min} e^{\frac{\rho_i - \rho_{\min}}{\rho_{\min}}} \le \rho_{\min} (1 + \tfrac{\rho_i - \rho_{\min}}{\rho_{\min}}) = \rho_i \le \max\nolimits_j \rho_j.
	 \end{equation*} 
	 The weighted average of the terms $\rho^{\text{eff}}_i$ in (\ref{eq:defNeg}) must also adhere to the bounds imposed on them, as was to be shown.
	 
	 \textit{Property 6:} Scale-invariance readily follows by substitution. For any $\rho_i \rightarrow \alpha \rho_i$ transformation with $\alpha > 0$, the normalized measures $\tilde{\rho}_i$ remain constant, and the minimal term becomes $\rho_{\min}\rightarrow \alpha \rho_{\min}$. Since the expression (\ref{eq:defNeg}) is linear in $\rho_{\min}$ (when treating $\tilde{\rho}_i$ independently), the desired scale-invariance property follows. 	
\end{proof}

\section{RESULTS} \label{section:results}

In this section, we compare the performance of various robustness measures in the context of learning a simple but instructive task. The task $\phi$ is for a single integrator robot $\x<\dot> = \u$, with $\norm[2]{\u} \le 1$, to \textit{always eventually} visit two nearby regions every $4$s until the time horizon $T = 10$s. The two goal regions are circular with radius $r_g = 0.2$ and are centered at $\x[g1] = [1.5\ 2.5]\tp$ and $\x[g2] = [2.5\ 1.5]\tp$. The robot itself starts from $\x[0] = [2.0\ 2.0]\tp$ and has an input constraint $\norm[2]{\u} \le 1$. We aim for a robustness of at least $\rho^{\phi} \ge 0.05$, as well as to minimize $C(\tau) = \int_{0}^{T} \norm[2]{\u(t)}^2 \diff t$. The scenario is simulated for $T = 10$s with a time step $\Delta t = 0.02$s.

The task is formulated as $\phi = \always[0][6]\left(\eventually[0][4] \mu_1 \and \eventually[0][4] \mu_2 \right)$, where $\mu_1 = (r_g - \norm{\x - \x[g1]} > 0)$ and $\mu_2 =  (r_g - \norm{\x - \x[g2]} > 0)$. The minimal robustness requirement implies that the distance to travel from one region to the other is 1.11, slightly more than what is physically possible for the robot in one second. For this reason, out of two potential solutions illustrated in Figure \ref{fig:setup}a, only the green one is feasible, and its cost can be calculated to be $C_{\text{opt}} = 2.02$.

We employ the guided \ac{PI2} method from Section \ref{section:PI2} to solve the described scenario. The algorithm allows the definition of initial $\gamma_i(t)$ funnels to guide the exploration towards the optimal solution. In particular, we examine three imposed guides for reaching the target regions, as depicted in Figure \ref{fig:setup}b. The three cases are referred to as \textit{strong}, \textit{weak}, and \textit{no guidance}, respectively. The figure shows the $\gamma_1(t)$ guides for reaching the goal region 1; the funnels for the second are defined from $\frac{\gamma_1(t) + \gamma_2(t)}{2} := \gamma_1(0)$ as $\gamma_2(t) = 2\gamma_1(0) - \gamma_1(t)$.

We compare the traditional, AG, and newly proposed robustness metrics in terms of their performance for guiding learning in context of the outlined scenario. The goal is to achieve task satisfaction as quickly and consistently as possible. To this end, the case scenario is solved 25 times using \ac{PI2} for the different configurations of guidance funnels and robustness metrics. Table II summarizes the percentage of finding feasible solutions as opposed to infeasible ones in each case. Figures 3-4 show the convergence of the task robustness measure $\rho^\phi$ and the achieved cost $C(\tau)$ as a function of the \ac{PI2} iteration number for the case of the successful runs.

Examining the figures, it is clear that the newly defined robustness measure surpasses both the traditional and AG metric in terms of accelerating the learning process. Surprisingly, the latter behaves quite poorly, most likely due to the learning method not being suited for handling discontinuous costs well. It is important to note here that the metrics and thus costs are not comparable in terms of their values, but in terms of when the task becomes satisfied due to the soundness property. The new metric consistently achieves task satisfaction 20-25\% faster than the traditional metric.

It is also worth discussing the achieved results in terms of which local minima was found by the algorithm, as summarized by Table II. In case of strong guidance, all 3 robustness metrics mainly converge to the true solution, although the AG metric sometimes still snaps into an infeasible one even in this case. With no guidance, the traditional metric is clearly superior in terms of finding the correct solution, although it still converges slower than the new metric in case the latter finds it as well. With a minimal weak initial nudge in the optimal direction, both the traditional and new metrics converge to the true optimum in $100\%$ of the runs, whereas the AG metric does not exhibit such a trend; in fact, the opposite seems to be the case. We conjecture that this behavior is caused by the AG metric being too rewarding for increases in robustness metrics, as seen through Figure \ref{fig:newcurves}. Conversely, the new metric rewards small increases of terms in a conjunction more than very large ones, leading to an arguably more desired convergence behavior of all the terms rising more or less together. We emphasize that the AG metric was designed to be rewarding in order to achieve higher robustness against noise than the traditional metric.

\begin{table}[t]
	\centering
	\caption{Percentage of \ac{PI2} runs for which feasible solutions were found for the case scenario. The results reflect 25 randomized runs for each configuration of robustness measure and guidance used to solve the problem.}
	\label{tab:results}
	
	\resizebox{0.7\linewidth}{!}{%
		\begin{tabular}{|l|c|c|c|}
			\hline
			\multicolumn{1}{|c|}{\multirow{2}{*}{Robustness metric}} & \multicolumn{3}{c|}{Feedback guidance} \\ \cline{2-4} 
			\multicolumn{1}{|c|}{}                                   & none       & weak        & strong      \\ \hline
			traditional \cite{donze2010robust}                       & 84\%       & 100\%       & 100\%       \\ \hline
			AG \cite{mehdipour2019arithmetic}                        & 48\%       & 28\%        & 92\%        \\ \hline
			new                                                      & 60\%       & 100\%       & 100\%       \\ \hline
		\end{tabular}%
	}\vspace{-4mm}
\end{table}

A final observation is that the proposed metric is able to achieve improved convergence rates without deviating much from the traditional robustness metric. This is seen from Figure \ref{fig:newcurves}; the parameter $\nu = 3.0$ was used in the simulation examples. Thus, the new metric expresses a similar robustness interpretation as the traditional metric. This is opposed to the AG metric, which expresses the (not necessarily true) desire to reach and stay at positive robustness states faster and for a longer period of time. These findings emphasize an additional possibility of the newly defined metric: $\nu$ allows the user to control how close it approximates the traditional metric. Varying its value throughout the learning procedure may lead to further improved results.

\begin{figure}[h]
	\begin{tikzpicture}%
	\captionsetup{format=plain,labelformat=parens,labelsep=space}
	\setcounter{subfigure}{0}
	\begin{groupplot}[group style={group name=myplots,group size= 1 by 2, vertical sep = 3mm},width=0.5\linewidth]
		
		\nextgroupplot[grid=both,width=0.93\linewidth,xmin=0,ymin=-0.6,ymax=0.6,xmax=10,xlabel=$t$,legend pos=north west,axis x line=bottom, axis y line = left,ylabel near ticks,ylabel style={rotate=-90},minor x tick num=1, minor y tick num = 1, enlarge x limits={abs=0.05}, enlarge y limits={abs=0.2},grid style={line width=.2pt, draw=gray!50}, major grid style={line width=.4pt,draw=gray!80},yscale=0.6,legend style={font=\small}]
		
		\addplot [name path=upper,draw=none] coordinates {(0,0.707) (10,0.707)};
		\addplot [name path=lower,draw=none] coordinates {(0,0.507) (10,0.507)};
		\addplot [fill=green, opacity=0.3] fill between[of=upper and lower];	
		\addplot [name path=lower,draw=none] coordinates {(0,-0.707) (10,-0.707)};
		\addplot [name path=upper,draw=none] coordinates {(0,-0.507) (10,-0.507)};
		\addplot [fill=green, opacity=0.3] fill between[of=upper and lower];		
		
		\addplot [very thick, green!70!black, mark=.] coordinates {(0,0) (2,0.557) (4,-0.557) (6,0.557) (8,-0.557) (10,-0.557)};
		\addplot [very thick, red, mark=.] coordinates {(0,0) (3,0.557) (4,-0.557) (6,-0.557) (7,0.557) (10,0.557)};
		
		\node at (axis cs:5,0.6) {\color{green!50!black} \textbf{goal 1}};
		\node at (axis cs:5,-0.67) {\color{green!50!black} \textbf{goal 2}};
		
		\nextgroupplot[grid=both,width=0.93\linewidth,xmin=0,ymin=-2.0,ymax=0.1,xmax=10,xlabel=$t$,ylabel=$\rho^{\mu_1}$,legend pos=north west,axis x line=bottom, axis y line = left,ylabel near ticks,ylabel style={rotate=-90},minor x tick num=1, minor y tick num = 1, enlarge y limits={abs=0.2},grid style={line width=.2pt, draw=gray!50}, major grid style={line width=.4pt,draw=gray!80},yscale=0.6,legend style={font=\small}, ylabel shift=-5mm]
		
		\addplot [name path=upper,draw=none] coordinates {(0,0.2) (10,0.2)};
		\addplot [name path=lower,draw=none] coordinates {(0,0) (10,0)};
		\addplot [fill=green, opacity=0.3] fill between[of=upper and lower];	
		\addplot [very thick, green!70!black, mark=.] coordinates {(0,-0.507) (2,0.05) (4,-1.064) (6, 0.05) (8,-1.064) (10,-1.064)};
		\addplot [ultra thick, blue!80!black, mark=.] coordinates {(0,-0.75) (2,-0.3) (4,-1.2) (6,-0.3) (8,-1.2) (10,-1.2)};
		\addplot [ultra thick, blue!50!black, mark=.] coordinates {(0,-1.5) (2,-1.3) (4,-1.7) (6,-1.3) (8,-1.7) (10,-1.7)};
		\addplot [ultra thick, blue!20!black, mark=.] coordinates {(0,-2) (10,-2)};
		
		\draw [thick,blue!60!black,->] (axis cs: 0.1,-0.8) to[bend left] (axis cs:0.2,-1.9);
		\node at (axis cs:1.6,-0.85) {\color{blue!80!black} \textbf{\small strong guide}};
		\node at (axis cs:1.85,-1.15) {\color{blue!50!black} \textbf{\small weak guide}};
		\node at (axis cs:1.5,-1.85) {\color{blue!20!black} \textbf{\small no guide}};
		
		\draw [thick,green!70!black,->] (axis cs: 7.9,-0.2) to (axis cs:7,-0.45);
		\node at (axis cs:9,-0.15) {\color{green!70!black} \textbf{\small optimal}};
		\node at (axis cs:9,-0.4) {\color{green!70!black} \textbf{\small trajectory}};
		\node at (axis cs:5, 0.105) {\color{green!50!black} \textbf{goal 1}};
	\end{groupplot}
	\node[text width=9cm,align=center,anchor=north] at ([yshift=-4mm]myplots c1r1.south) {\captionof{subfigure}{\small Out of two potential local minimum solutions for satisfying the task, the green is feasible and optimal with respect to minimizing the input energy. The red is not feasible due to  actuator constraints.}};
	\node[text width=9cm,align=center,anchor=north] at ([yshift=-4mm]myplots c1r2.south) {\captionof{subfigure}{\small Evolution of the robustness measure $\rho^{\mu_1}(\x(t))$ for the optimal trajectory (green), and three $\gamma_1(t)$ guide funnels aiming to enforce it in a progressively more aggressive manner (blue).}};
	\end{tikzpicture}	
	\caption{Illustration of potential solutions and guides to satisfying the task outlined for the case study. The robot needs to visit both goal regions once every 4 seconds throughout the first 6 seconds of its motion.}
	\label{fig:setup}
\end{figure}

\vspace{-3mm}
\begin{figure*}[t]
	\centering
	\begin{tikzpicture}
	\pgfplotstableread{data/strong_guide_success_C.dat}\strongGuideSuccessC;
	\pgfplotstableread{data/weak_guide_success_C.dat}\weakGuideSuccessC;
	\pgfplotstableread{data/no_guide_success_C.dat}\noGuideSuccessC;
	\pgfplotstableread{data/strong_guide_success_rho.dat}\strongGuideSuccessRho;
	\pgfplotstableread{data/weak_guide_success_rho.dat}\weakGuideSuccessRho;
	\pgfplotstableread{data/no_guide_success_rho.dat}\noGuideSuccessRho;
	\captionsetup{format=plain,labelformat=parens,labelsep=space}
	\setcounter{subfigure}{0}
	\begin{groupplot}[group style={group name=myplots,group size= 3 by 2,horizontal sep=4mm,vertical sep=-11mm},width=\linewidth]
	
	\nextgroupplot[grid=both,width=0.39\linewidth,xmin=0,ymin=0.0,ymax=4.5,xmax=100,ylabel=$C$,axis x line=bottom, axis y line = left,ylabel near ticks,ylabel style={rotate=-90},xticklabels={,,},minor x tick num=1, minor y tick num = 1, enlarge y limits={abs=0.02},grid style={line width=.2pt, draw=gray!50}, major grid style={line width=.4pt,draw=gray!80},yscale=0.7,legend style={font=\small},domain=0.4:3.0]
	\addplot [very thick, lightblue, mark=.] table[x=k,y=C_base] {\strongGuideSuccessC}; \label{plots:resmin}
	\addplot [very thick, lightred, mark=.] table[x=k,y=C_ag] {\strongGuideSuccessC}; \label{plots:resAG}
	\addplot [very thick, darkgreen, mark=.] table[x=k,y=C_new] {\strongGuideSuccessC}; \label{plots:resnu3}
	
	\addplot [name path=upper,draw=none] table[x=k,y expr=\thisrow{C_base_max}] {\strongGuideSuccessC};
	\addplot [name path=lower,draw=none] table[x=k,y expr=\thisrow{C_base_min}] {\strongGuideSuccessC};
	\addplot [fill=lightblue, opacity=0.3] fill between[of=upper and lower];		
	\addplot [name path=upper,draw=none] table[x=k,y expr=\thisrow{C_ag_max}] {\strongGuideSuccessC};
	\addplot [name path=lower,draw=none] table[x=k,y expr=\thisrow{C_ag_min}] {\strongGuideSuccessC};
	\addplot [fill=lightred, opacity=0.3] fill between[of=upper and lower];		
	\addplot [name path=upper,draw=none] table[x=k,y expr=\thisrow{C_new_max}] {\strongGuideSuccessC};
	\addplot [name path=lower,draw=none] table[x=k,y expr=\thisrow{C_new_min}] {\strongGuideSuccessC};
	\addplot [fill=darkgreen, opacity=0.3] fill between[of=upper and lower];
	
	\nextgroupplot[grid=both,width=0.39\linewidth,xmin=0,ymin=0.0,ymax=4.5,xmax=100,axis x line=bottom, axis y line = left,yticklabels={,,},xticklabels={,,},minor x tick num=1, minor y tick num = 1, enlarge y limits={abs=0.02},grid style={line width=.2pt, draw=gray!50}, major grid style={line width=.4pt,draw=gray!80},yscale=0.7,legend style={font=\small},domain=0.0:4.5]
	\addplot [very thick, lightblue, mark=.] table[x=k,y=C_base] {\weakGuideSuccessC};
	\addplot [very thick, lightred, mark=.] table[x=k,y=C_ag] {\weakGuideSuccessC};
	\addplot [very thick, darkgreen, mark=.] table[x=k,y=C_new] {\weakGuideSuccessC};
	
	\addplot [name path=upper,draw=none] table[x=k,y expr=\thisrow{C_base_max}] {\weakGuideSuccessC};
	\addplot [name path=lower,draw=none] table[x=k,y expr=\thisrow{C_base_min}] {\weakGuideSuccessC};
	\addplot [fill=lightblue, opacity=0.3] fill between[of=upper and lower];		
	\addplot [name path=upper,draw=none] table[x=k,y expr=\thisrow{C_ag_max}] {\weakGuideSuccessC};
	\addplot [name path=lower,draw=none] table[x=k,y expr=\thisrow{C_ag_min}] {\weakGuideSuccessC};
	\addplot [fill=lightred, opacity=0.3] fill between[of=upper and lower];		
	\addplot [name path=upper,draw=none] table[x=k,y expr=\thisrow{C_new_max}] {\weakGuideSuccessC};
	\addplot [name path=lower,draw=none] table[x=k,y expr=\thisrow{C_new_min}] {\weakGuideSuccessC};
	\addplot [fill=darkgreen, opacity=0.3] fill between[of=upper and lower];	
	
	\nextgroupplot[grid=both,width=0.39\linewidth,xmin=0,ymin=0.0,ymax=4.5,xmax=100,axis x line=bottom, axis y line = left,yticklabels={,,}, xticklabels={,,},minor x tick num=1, minor y tick num = 1, enlarge y limits={abs=0.02},grid style={line width=.2pt, draw=gray!50}, major grid style={line width=.4pt,draw=gray!80},yscale=0.7,legend style={font=\small},domain=0.0:4.5]
	\addplot [very thick, lightblue, mark=.] table[x=k,y=C_base] {\noGuideSuccessC};
	\addplot [very thick, lightred, mark=.] table[x=k,y=C_ag] {\noGuideSuccessC};
	\addplot [very thick, darkgreen, mark=.] table[x=k,y=C_new] {\noGuideSuccessC};
	
	\addplot [name path=upper,draw=none] table[x=k,y expr=\thisrow{C_base_max}] {\noGuideSuccessC};
	\addplot [name path=lower,draw=none] table[x=k,y expr=\thisrow{C_base_min}] {\noGuideSuccessC};
	\addplot [fill=lightblue, opacity=0.3] fill between[of=upper and lower];		
	\addplot [name path=upper,draw=none] table[x=k,y expr=\thisrow{C_ag_max}] {\noGuideSuccessC};
	\addplot [name path=lower,draw=none] table[x=k,y expr=\thisrow{C_ag_min}] {\noGuideSuccessC};
	\addplot [fill=lightred, opacity=0.3] fill between[of=upper and lower];		
	\addplot [name path=upper,draw=none] table[x=k,y expr=\thisrow{C_new_max}] {\noGuideSuccessC};
	\addplot [name path=lower,draw=none] table[x=k,y expr=\thisrow{C_new_min}] {\noGuideSuccessC};
	\addplot [fill=darkgreen, opacity=0.3] fill between[of=upper and lower];	
	
	\nextgroupplot[grid=both,width=0.39\linewidth,xmin=0,ymin=-0.5,ymax=0.05,xmax=100,xlabel=$k$,ylabel=$\rho^{\phi}$,axis x line=bottom, axis y line = left,ylabel near ticks,ytick={-0.5,0},ylabel shift = -15pt,ylabel style={rotate=-90},minor x tick num=1, minor y tick num = 4, enlarge y limits={abs=0.02},grid style={line width=.2pt, draw=gray!50}, major grid style={line width=.4pt,draw=gray!80},yscale=0.7]
	\addplot [very thick, lightblue, mark=.] table[x=k,y=rho_base] {\strongGuideSuccessRho};
	\addplot [very thick, lightred, mark=.] table[x=k,y=rho_ag] {\strongGuideSuccessRho};
	\addplot [very thick, darkgreen, mark=.] table[x=k,y=rho_new] {\strongGuideSuccessRho};
	
	\addplot [name path=upper,draw=none] table[x=k,y expr=\thisrow{rho_base_max}] {\strongGuideSuccessRho};
	\addplot [name path=lower,draw=none] table[x=k,y expr=\thisrow{rho_base_min}] {\strongGuideSuccessRho};
	\addplot [fill=lightblue, opacity=0.3] fill between[of=upper and lower];		
	\addplot [name path=upper,draw=none] table[x=k,y expr=\thisrow{rho_ag_max}] {\strongGuideSuccessRho};
	\addplot [name path=lower,draw=none] table[x=k,y expr=\thisrow{rho_ag_min}] {\strongGuideSuccessRho};
	\addplot [fill=lightred, opacity=0.3] fill between[of=upper and lower];		
	\addplot [name path=upper,draw=none] table[x=k,y expr=\thisrow{rho_new_max}] {\strongGuideSuccessRho};
	\addplot [name path=lower,draw=none] table[x=k,y expr=\thisrow{rho_new_min}] {\strongGuideSuccessRho};
	\addplot [fill=darkgreen, opacity=0.3] fill between[of=upper and lower];	
	
	\nextgroupplot[grid=both,width=0.39\linewidth,xmin=0,ymin=-0.5,ymax=0.05,xmax=100,xlabel=$k$,axis x line=bottom, axis y line = left,yticklabels={,,},minor x tick num=1, minor y tick num = 1, enlarge y limits={abs=0.02},grid style={line width=.2pt, draw=gray!50}, major grid style={line width=.4pt,draw=gray!80},yscale=0.7]
	\addplot [very thick, lightblue, mark=.] table[x=k,y=rho_base] {\weakGuideSuccessRho};
	\addplot [very thick, lightred, mark=.] table[x=k,y=rho_ag] {\weakGuideSuccessRho};
	\addplot [very thick, darkgreen, mark=.] table[x=k,y=rho_new] {\weakGuideSuccessRho};

	\addplot [name path=upper,draw=none] table[x=k,y expr=\thisrow{rho_base_max}] {\weakGuideSuccessRho};
	\addplot [name path=lower,draw=none] table[x=k,y expr=\thisrow{rho_base_min}] {\weakGuideSuccessRho};
	\addplot [fill=lightblue, opacity=0.3] fill between[of=upper and lower];		
	\addplot [name path=upper,draw=none] table[x=k,y expr=\thisrow{rho_ag_max}] {\weakGuideSuccessRho};
	\addplot [name path=lower,draw=none] table[x=k,y expr=\thisrow{rho_ag_min}] {\weakGuideSuccessRho};
	\addplot [fill=lightred, opacity=0.3] fill between[of=upper and lower];		
	\addplot [name path=upper,draw=none] table[x=k,y expr=\thisrow{rho_new_max}] {\weakGuideSuccessRho};
	\addplot [name path=lower,draw=none] table[x=k,y expr=\thisrow{rho_new_min}] {\weakGuideSuccessRho};
	\addplot [fill=darkgreen, opacity=0.3] fill between[of=upper and lower];	
	
	\nextgroupplot[grid=both,width=0.39\linewidth,xmin=0,ymin=-0.5,ymax=0.05,xmax=100,xlabel=$k$,axis x line=bottom, axis y line = left,yticklabels={,,},minor x tick num=1, minor y tick num = 1, enlarge y limits={abs=0.02},grid style={line width=.2pt, draw=gray!50}, major grid style={line width=.4pt,draw=gray!80},yscale=0.7]
	\addplot [very thick, lightblue, mark=.] table[x=k,y=rho_base] {\noGuideSuccessRho};
	\addplot [very thick, lightred, mark=.] table[x=k,y=rho_ag] {\noGuideSuccessRho};
	\addplot [very thick, darkgreen, mark=.] table[x=k,y=rho_new] {\noGuideSuccessRho};
	
	\addplot [name path=upper,draw=none] table[x=k,y expr=\thisrow{rho_base_max}] {\noGuideSuccessRho};
	\addplot [name path=lower,draw=none] table[x=k,y expr=\thisrow{rho_base_min}] {\noGuideSuccessRho};
	\addplot [fill=lightblue, opacity=0.3] fill between[of=upper and lower];		
	\addplot [name path=upper,draw=none] table[x=k,y expr=\thisrow{rho_ag_max}] {\noGuideSuccessRho};
	\addplot [name path=lower,draw=none] table[x=k,y expr=\thisrow{rho_ag_min}] {\noGuideSuccessRho};
	\addplot [fill=lightred, opacity=0.3] fill between[of=upper and lower];		
	\addplot [name path=upper,draw=none] table[x=k,y expr=\thisrow{rho_new_max}] {\noGuideSuccessRho};
	\addplot [name path=lower,draw=none] table[x=k,y expr=\thisrow{rho_new_min}] {\noGuideSuccessRho};
	\addplot [fill=darkgreen, opacity=0.3] fill between[of=upper and lower];	
	\end{groupplot}
	\node[text width=6cm,align=center,anchor=north] at ([yshift=-8mm]myplots c1r2.south) {\captionof{subfigure}{strong guidance}};
	\node[text width=6cm,align=center,anchor=north] at ([yshift=-8mm]myplots c2r2.south) {\captionof{subfigure}{weak guidance}};
	\node[text width=6cm,align=center,anchor=north] at ([yshift=-8mm]myplots c3r2.south) {\captionof{subfigure}{no guidance}};
	\path (myplots c1r1.north west|-current bounding box.north)--
	coordinate(legendpos)
	(myplots c3r1.north east|-current bounding box.north);
	\matrix[
	matrix of nodes,
	anchor=south,
	draw,
	inner sep=0.1em,
	draw
	]at([yshift=1ex, xshift=-2mm]legendpos)
	{
		\ref{plots:resmin}& \small traditional &[3pt]
		\ref{plots:resAG}& \small AG &[3pt]
		\ref{plots:resnu3}& \small new ($\nu = 3$) \\};
	\end{tikzpicture}
	\caption{Convergence of the cost and robustness metrics throughout the \ac{PI2} algorithm running under strong, weak, or no guidance. The results show the distribution of successful runs obtained from 25 sample runs of the algorithm, for the discussed traditional, AG, and newly proposed robustness metrics. The median task robustness measure $\rho^{\phi}$ and achieved cost $C$ are plotted as a function of the \ac{PI2} iteration number $k$. All results, excluding the top and bottom 10\%, lie in the shaded regions. The newly defined metric always converges to near the optimal solution and achieves task satisfaction most quickly and with least variance. The optimal cost can be calculated as $C_{\text{opt}} = 2.02$.}
	\label{fig:resultsSuccess}
\end{figure*}

\addtolength{\textheight}{0cm}   

\section{CONCLUSIONS} \label{section:conclusions}

In this paper, we presented theoretical results regarding the form of possible robustness metrics for quantifying the satisfaction of \ac{STL} tasks. The findings motivated the definition of a sample new robustness metric, whose improved performance for accelerating learning was demonstrated through a simulation case study. Our preliminary results are promising and motivate further research into the potential of defining robustness metrics for their role in general-purpose reward shaping. Further work on the topic includes a more rigorous exploration of the properties of potential robustness metrics and their imposed restrictions on the metric itself, as well as conducting a thorough simulation study to verify their superior performance across a wider spectrum of scenarios.




%
%
%

\bibliographystyle{IEEEtran}
{\renewcommand{\baselinestretch}{0.98}
\bibliography{IEEEabrv,MyBib_conf}

\end{document}